\documentclass[onecolumn,draft, 11pt]{IEEEtran}

\usepackage{color}

\usepackage{bm,graphicx,tabularx,array,geometry,amsmath,amsthm,thmtools}

\usepackage{mathtools}

\usepackage{caption}
\usepackage{amsfonts}
\usepackage{bbm}
\usepackage{makecell}
\usepackage{multirow}
 \usepackage{amssymb}
\usepackage{amsthm}
\usepackage{txfonts}
\usepackage[T1]{fontenc}
\usepackage{tikz}
\usepackage[scr=dutchcal]{mathalfa}
\let\mathscr\mathbscr
\usepackage[normalem]{ulem}

\newcolumntype{x}[1]{>{\centering\arraybackslash}p{#1}}

\declaretheoremstyle[headfont=\bfseries, 
    bodyfont=\normalfont]{normalhead}

\allowdisplaybreaks 

\usepackage[]{algorithm2e}
\newtheorem{Theorem}{Theorem}

\newtheorem{Lemma}{Lemma}

\newtheorem{Definition}{Definition}

\begin{document}

\title{ Seeded Graph Matching: Efficient Algorithms and Theoretical Guarantees\footnote{This work was presented in the $51^{st}$ Annual Asilomar Conference on Signals, Systems, and Computers.}}


\author{Farhad Shirani, Siddharth Garg, and Elza Erkip
\\
\{fsc265,siddharth.garg,elza\}@nyu.edu\\
Department of Electrical and Computer Engineering\\
New York University, New York, New York, 11201 \\\date{} }

\maketitle

\begin{abstract}
\hspace{0.01in}
In this paper, a new information theoretic framework for graph matching is introduced. Using this framework, the graph isomorphism and seeded graph matching problems are studied. 
The maximum degree algorithm for graph isomorphism is analyzed and sufficient conditions for successful matching are rederived using type analysis. Furthermore, a new seeded matching algorithm with polynomial time complexity is introduced. The algorithm uses `typicality matching' and techniques from point-to-point communications for reliable matching.
Assuming an Erd\"os-R\'enyi model on the correlated graph pair, it is shown that successful matching is guaranteed when the number of seeds grows logarithmically with the number of vertices in the graphs. The logarithmic coefficient is shown to be inversely proportional to the mutual information between the edge variables in the two graphs. 
%
%
%
%
\end{abstract} 


%
\IEEEpeerreviewmaketitle

\section{Introduction}
Network graphs emerge naturally in modeling a wide range of phenomena including
social interactions, database systems, biological systems, and epidemiology. In many applications
such as DNA sequencing, pattern recognition, and image processing, it is desirable to find algorithms
to match correlated network graphs. In other applications, such as social networks and database
systems, privacy considerations require the network operators to enforce security safeguards in order to preclude de-anonymization using graph matching.

In this work, we consider two formulations of the graph matching problem: i) \textbf{Graph Isomorphism:} In this formulation it is assumed that an agent is given two identical copies of the same graph. The vertices in the first graph $g^1$ are labeled. Whereas the second graph, $g^2$, is an unlabeled graph. The objective is to find the natural labeling of the vertices in $g^2$ based on the labels in $g^1$ \cite{iso1,iso2,iso3,iso4}.
\textbf{ii) Seeded Graph Matching:}  In this formulation, $g^1$ and $g^2$ are two statistically correlated graphs.  Here, it is assumed that the agent is provided with additional side-information in the form of seeds. More precisely, for a small subset of vertices, the correct labeling across both graphs $g^1$ and $g^2$ is given \cite{seed1,seed2,seed3,seed4,efe,kiavash,Grossglauser}. This is shown in Figure \ref{fig:corr_graph}. 
There are various practical scenarios captured in this formulation. One pertinent application is the de-anonymization of users who are members of multiple social networks. Currently, many web users are members of multiple online social networks such as Facebook, Twitter, Google+, LinkedIn, etc..  Each online network represents a subset of the users' ``real" ego-networks. Graph matching provides algorithms to de-anonymize the users by reconciling these online network graphs, that is, to identify all the accounts belonging to the same individual. In this context, the availability of seeds is justified by the fact that a small fraction of individuals explicitly link their accounts across multiple networks. In this case, these linked accounts can be used as seeds in the agent's matching algorithm. It turns out, that in many cases, the agent may leverage these connections to identify a very large fraction of the users in the network ~\cite{seed1,seed2,seed3,seed4,kiavash}.
\begin{figure}
\centering 
\includegraphics[width=3in,draft=false]{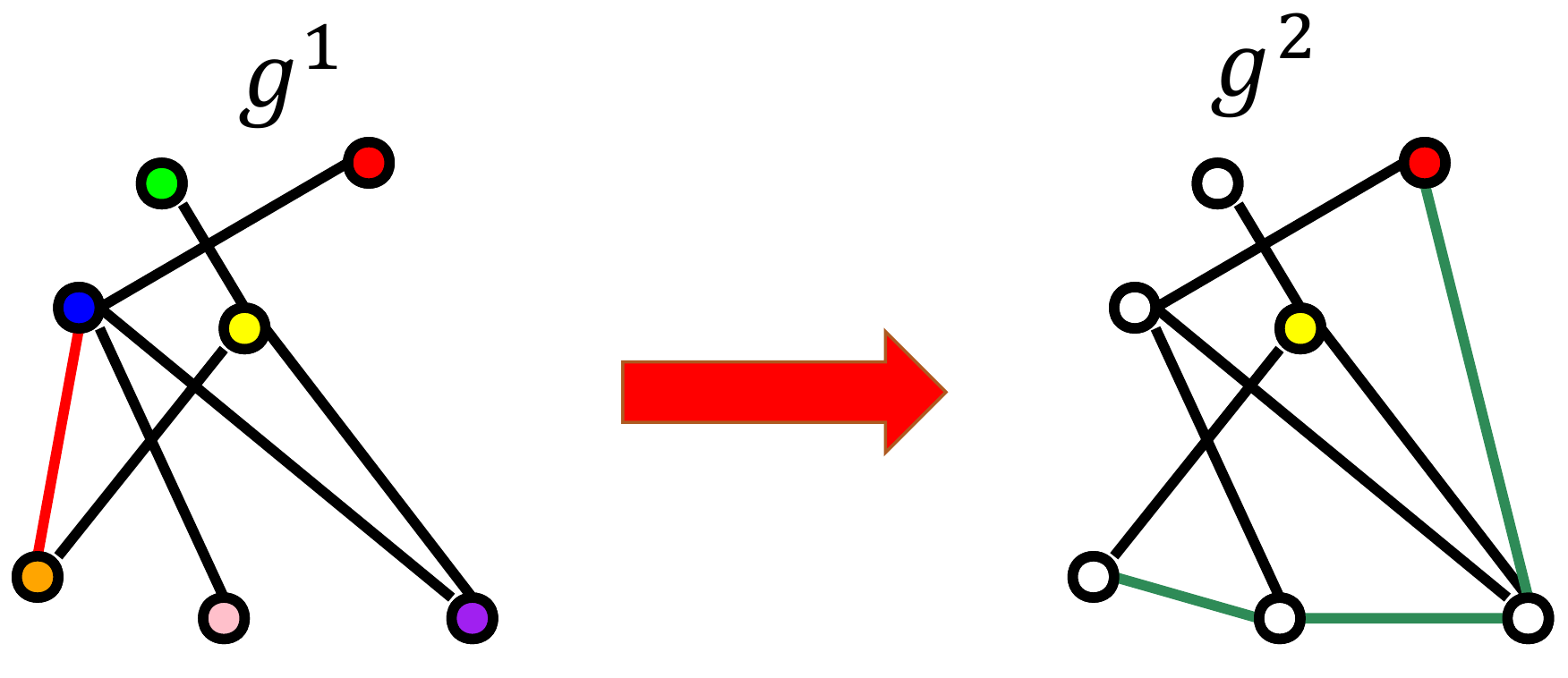}
\caption{The graph pair $(g^1,g^2)$ are correlated. The objective is to restore the labels in graph $g^2$ based on the labels in $g^1$ and the seed vertices.
}
\label{fig:corr_graph}
\end{figure}

 Graph matching is a classical problem which is of significant interest in multiple disciplines. As a result, there is a large body of work considering the problem under varying assumptions on the graph statistics. The graph isomorphism problem was studied in \cite{iso1,iso2,iso3,iso4}. Under the Erd\"os- R\'enyi graph model, tight necessary and sufficient conditions were derived \cite{ER,wright} and polynomial time algorithms were proposed \cite{iso1,iso2,iso4}. The problem of matching of correlated graphs was considered for pairs of Erd\"os- R\'enyi graphs \cite{corr1,corr2,corr3,corr4}, as well as graphs with community structure \cite{corr_com_1,kiavash,efe,Grossglauser}. Furthermore, conditions for seeded and seedless graph matching were derived for general graph scenarios \cite{seed1,seed2,seed3,seed4}. While great progress has been made in determining the conditions for successful matching and providing feasible matching algorithms, there are shortcomings in the current approaches. Many of the results provide solutions which assume unbounded computational resources (i.e. existence results).
Furthermore, performance characterizations are often provided in the form of complex algebraic characterizations which lead to limited insight for application in practical scenarios. In this work, we propose a new \textit{information theoretic} framework for the study of the fundamental limits of graph matching as well as for the design of feasible algorithms. We use the new framework to study the graph isomorphism problem as well as seeded graph matching. 

The contributions can be summarized as follows. In graph isomorphism, we use tools from information theory such as concentration of measure properties and type analysis to analyze  the sufficient conditions for successful matching for the maximum degree algorithm studied in \cite{iso1,iso2}. We  rederive  the conditions for successful matching given in \cite{iso2} under the Erd\"os-R\'enyi graph model.
 In seeded matching, we propose a new matching algorithm which is inspired by a fingerprinting strategy introduced recently in \cite{Allerton}. The new algorithm  makes use of information theoretic tools such as typicality as well as ideas from point-to-point communications. This strategy has polynomial-time complexity and guarantees reliable matching under the Erd\"os-R\'enyi graph model as long as the seed size $\Lambda$ satisfies  $$\Lambda_n> \Theta\left(\frac{\log{n}}{I(X_1;X_2)}\right),$$ 
where $I(X_1;X_2)$ represents the mutual information between the edge distributions in $g^{1}$ and $g^{2}$ and $n$ is the number of vertices.
In contrast with~\cite{Beyah,kiavash}, which also consider the seeded de-anonymization problem, our results have a clear information theoretic interpretation. That is, each seed provides roughly $I(X_1;X_2)$ bits of information about the label of any unlabeled vertex. As a result, almost $\frac{\log{n}}{I(X_1;X_2)}$ seeds are necessary for successful matching.

The rest of the paper is organized as follows: In Section II we introduce our notation and problem formulation, Section III considers the graph isomorphism problem, Section IV investigates seeded graph matching, and Section V concludes the paper.

\section{Problem Formulation}
In this section, we construct the foundations of our study of the graph matching problem by providing a rigorous formalization of the problem. We assume a stochastic model where the graphs are generated randomly based on an underlying distribution. Particularly, the two graphs are produced based on the Erd\"os-R\'enyi model. Furthermore, the edges in the two graphs are statistically correlated through a joint probability distribution.  More precisely, each edge in the second graph is correlated with its corresponding edge in the first graph and is independent of all other edges. The first graph is labeled whereas the second graph is unlabeled. The objective is to determine the labeling of the second graph which conforms with the joint distribution between the edges of the two graphs. We proceed by presenting our formalization of the graph matching problem. 

\begin{Definition}
 An unlabeled graph $g$ consists of a pair $(\mathcal{V}_n,\mathcal{E}_n)$. The set $\mathcal{V}_n=\{v_{n,1},v_{n,2},\cdots,v_{n,n}\}$ is called the set of vertices of the graph, and the set $\mathcal{E}_n\subset\{(v_{n,i},v_{n,j})|i\in [1,n], j\in[1,n]\}$ is called the set of edges of the graph. For a given vertex $v_{n,i}, i\in [1,n]$ the set of neighboring vertices is defined as $\mathcal{E}_{n,i}=\{v_{n,j}|(v_{n,i},v_{n,j})\in \mathcal{E}_n\}$. The degree of $v_{n,i}$ is defined as $d_{v_{n,i}}=|\mathcal{E}_{n,i}|$.
\end{Definition}

\begin{Definition}
  For an unlabeled graph $g=(\mathcal{V}_n,\mathcal{E}_n)$, a labeling is defined as a bijective function $\sigma: \mathcal{V}_n\to [1,n]$.  
\end{Definition}

\begin{Definition}
A labeled graph $\tilde{g}$ consists of a pair $(g, \sigma)$, where $g$ is an unlabeled graph and $\sigma$ is a labeling for $g$. 
\end{Definition}

\begin{Definition}
For a given $p\in [0,1]$, an Erd\"os-R\'enyi (ER) graph $g_{n,p}$ is a randomly generated unlabeled graph with vertex set $\mathcal{V}_n$ and edge set $\mathcal{E}_n$, such that
\begin{align*}
 Pr(g_{n,p})= p^{|\mathcal{E}_n|}(1-p)^{\left({n\choose 2}- |\mathcal{E}_n|\right)}.
\end{align*}

\end{Definition}

In this work, we consider families of correlated pairs of labeled Erd\"os-R\'enyi graphs $\underline{g}_{n,P_{n,X_1,X_2}}=(\tilde{g}^1_{n,p_{n,1}},\tilde{g}^2_{n,p_{n,2}})$. 

\begin{Definition}
For a given joint distribution $P_{X_1,X_2}$ on binary random variables $X_1$ and $X_2$, a correlated pair of labeled Erd\"os-Renyi (CER) graphs $\tilde{\underline{g}}_{n,P_{X_1,X_2}}=(\tilde{g}^1_{n,p_{1}},\tilde{g}^2_{n,p_{2}})$ is characterized by: i) the pair of Erd\"os-Renyi graphs $g^i_{n,p_{i}}, i\in \{1,2\}$, ii) the pair of labelings $\sigma^i$ for the unlabeled graphs $g^i_{n,p_{1}}, i\in \{1,2\}$, and iii) the probability distribution $P_{X_1,X_2}(x_1,x_2), (x_1,x_2) \in \{0,1\}^2$, such that: 
\\1)The pair $g^i_{n,p_{i}}, i\in \{1,2\}$ have the same set of vertices $\mathcal{V}_n=\mathcal{V}_n^1=\mathcal{V}_n^2$. 
\\2) The marginals satisfy $P_{X_i}(X_i=1)=p_{i}, i\in \{1,2\}$. 
\\3) For any two edges $e^i=(v^i_{n,1},v^i_{n,2}), i\in \{1,2\}$, we have
 \begin{align*}
 &Pr\left(\mathbbm{1}(e^1\in \mathcal{E}_n^1) =  \alpha, \mathbbm{1}(e^2\in \mathcal{E}_n^2)=\beta)\right)=
\begin{cases}
P_{X_1,X_2}(\alpha,\beta),& \text{if } \sigma^1_{j}(v^1_{n,j})=\sigma^2_j(v^2_{n,j}), j\in \{1,2\}\\
P_{X_1}(\alpha)P_{X_2}(\beta), & \text{Otherwise}
\end{cases}
,\forall \alpha,\beta\in \{0,1\}^2,
\end{align*}
where $\mathbbm{1}(\cdot)$ is the indicator function. 
\end{Definition}

\begin{Definition}
For a given joint distribution $P_{X_1,X_2}$, a correlated pair of partially labeled Erd\"os-Renyi (CPER) graphs $\underline{g}_{n,P_{X_1,X_2}}=(\tilde{g}^1_{n,p_{1}},{g}^2_{n,p_{2}})$  is characterized by: i) the pair of Erd\"os-Renyi graphs $g^i_{n,p_{i}}, i\in \{1,2\}$, ii) a labeling $\sigma^1$ for the unlabeled graph $g^1_{n,p_{1}}$, and iii) the probability distribution $P_{X_1,X_2}(x_1,x_2), (x_1,x_2) \in \{0,1\}^2$,  such that there exists a labeling $\sigma^2$ for the graph $g^2_{n,p_{2}}$ for which $(\tilde{g}^1_{n,p_{1}},\tilde{g}^2_{n,p_{2}})$ is a CER, where $\tilde{g}^2_{n,p_{2}}\triangleq (g^2_{n,p_{2}},\sigma^2)$.
\end{Definition}


The following defines a matching algorithm:

\begin{Definition}
 A matching algorithm for the family of CPERs $\underline{g}_{n,P_{n,X_1,X_2}}=(\tilde{g}^1_{n,p_{n,1}},{g}^2_{n,p_{n,2}}), n\in \mathbb{N}$ is a sequence of functions 
 $f_n:\underline{g}_{n,P_{n,X_1,X_2}}\mapsto \hat{\sigma}_n^2$ such that $P(\sigma_n^2=\hat{\sigma}_n^2)\to 1$ as $n\to \infty$, where $\sigma^2_n$ is the labeling  for the graph $g^2_{n,p_{n,2}}$ for which $(\tilde{g}^1_{n,p_{n,1}},\tilde{g}^2_{n,p_{n,2}})$ is a CER, where $\tilde{g}^2_{n,p_{n,2}}\triangleq (g^2_{n,p_{n,2}},\sigma_n^2)$.

\end{Definition}
\begin{Definition}
 A matching algorithm for the family of CPERs $\underline{g}_{n,P_{n,X_1,X_2}}=(\tilde{g}^1_{n,p_{n,1}},{g}^2_{n,p_{n,2}}), n\in\mathbb{N}$ is called a graph isomorphism algorithm if $X_1=X_2$ with probability one.
\end{Definition}
\begin{Definition}
 A seeded matching algorithm for the family of CPERs $\underline{g}_{n,P_{n,X_1,X_2}}=(\tilde{g}^1_{n,p_{n,1}},{g}^2_{n,p_{n,2}})$
 and seed sets $\mathcal{S}_n\subset \mathcal{V}^2_n$
  is a sequence of functions 
 $f_n:\underline{g}_{n,P_{n,X_1,X_2}} \times \sigma_n^2\big|_{\mathcal{S}_n}\mapsto \hat{\sigma}_n^2$ such that $P(\sigma_n^2=\hat{\sigma}_n^2)\to 1$ as $n\to \infty$, where $\sigma_n^2\big|_{\mathcal{S}_n}$ is the restriction of $\sigma_n^2$ to the subset $\mathcal{S}_n$ of the vertices. $\Lambda_n=|\mathcal{S}_n|$ is called the seed size.
\end{Definition}
The following defines an achievable region for the graph matching problem.
\begin{Definition}
 For the graph matching problem, a family of sets of distributions $\widetilde{P}=(\mathcal{P}_n)_{n\in \mathbb{N}}$ is said to be in the achievable region if for every sequence of distributions $P_{n,X_1,X_2}\in \mathcal{P}_n, n\in \mathbb{N}$, there exists a matching algorithm.
 \end{Definition} 
\begin{Definition}
For the seeded graph matching problem, a family of pairs of distribution sets and seed sizes $\widetilde{P}_s=(\mathcal{P}_n, \Lambda_n)_{n\in \mathbb{N}}$ is said to be in the achievable region if for every sequence of pairs $P_{n,X_1,X_2} \in \mathcal{P}_n$ and $\mathcal{S}_n\in \mathcal{V}^2_n$ such that $\mathcal{S}_n=|\Lambda_n|$, there exists a seeded matching algorithm. 
\end{Definition}

\section{Graph Isomorphism}
In this section, we consider the graph isomorphism problem. First, we analyze the maximum degree algorithm (MDA) proposed in \cite{iso1,iso2} and re-derive the achievable distribution region using information theoretic quantities.

\subsection{Maximum Degree Algorithm}
In their seminal paper, Babai et. al. \cite{iso1} proposed the first polynomial time algorithm for graph isomorphism for families of CPERs with parameter $p_n=\frac{1}{2}, n\in \mathbb{N}$. This result was later extended to include a larger achievable region, namely the region characterized by  $p_n\in [\omega({{\log{n}}/{n^{\frac{1}{5}}}}),1-\omega({\log{n}}/{n^{\frac{1}{5}}})]$,  \cite{iso2} . The algorithm operates in two steps. In the first step,  the agent who is matching the two graphs exploits the uniqueness of the degrees of the highest degree vertices in the graphs to match these vertices. In the next step, these vertices are used as seeds to match the rest of the vertices.
 More precisely, the agent orders the vertices in both graphs based on their degrees. Let $d^{j}_{n,(i)}, j\in \{1,2\}$ be the $i$th largest degree in $g^j_{n,p_{n,j}}$. Also, let $k^j_{n,(i)}$ be the number of vertices in $g^j_{n,p_{n,j}}$ with degree  $d^{j}_{n,(i)}$ and let $v^{j}_{n,(i),k}, k\in [1,k^j_{n,(i)}]$ be the $k$th vertex with degree $k^j_{n,(i)}$. Define 
\begin{align*}
 d_{n,u}= argmin \{i| k^1_{n,(j)}=1, \forall j\geq i\}.
\end{align*}

In other words, if we list the vertices in $g^1$ in decreasing degree order and  start from the vertex with the highest degree, the first $d_{n,u}$ vertices have unique degrees. The following lemma provides a lower bound on $d_{n,u}$.

\begin{Lemma}[ \cite{iso2} Theorem 3.15]
Suppose $m=o\left(\left(\frac{p(1-p)n}{\log{n}}\right)^{\frac{1}{4}}\right)$. Then,
 \begin{align*}
P(d_{n,u}>m)\to 1 \qquad \text{as} \qquad n\to \infty.
\end{align*}
\label{lem:degree}
\end{Lemma}
The agent matches the first $d_{n,u}$ vertices in the two graphs. Formally, the agent sets $\hat{\sigma}^2(v^2_{n,(i),1})=\sigma^1(v^1_{n,(i),1}), i\in [1,d_{n,u}]$. Next, it constructs a signature vector for each of the unmatched vertices. For each vertex $v^j_{n,(i),k}, j\in \{1,2\}, i<d_{n,u}, k\in [1,k_i]$, its signature is defined as the binary vector  $\underline{F}^j_{n,(i),k}=(F^j_{n,(i),k}(1), F^j_{n,(i),k}(2),\cdots,F^j_{n,(i),k}(d_{n,u}))$.
which indicates its connections to the matched vertices:
\begin{align*}
{F}^j_{n,(i),k}(l)=
\begin{cases}
 1 \qquad & \text{if} \qquad (v^j_{n,(i),k}, v^j_{n,(l),1})\in \mathcal{E}^j_{n}\\
 0& \text{Otherwise}
\end{cases},
\qquad l\in [1,d_{n,u}].
\end{align*}

For any arbitrary vertex $v^1_{n,(i),k}$, if there exists a unique vertex $v^2_{n,(i),k'}$ such that $\underline{F}^1_{n,(i),k}$=$\underline{F}^2_{n,(i),k'}$, then the agent matches the two vertices by setting $\hat{\sigma}^2(v^2_{n,(i),k})=\sigma^1(v^1_{n,(i),k'})$. If all of the vertices are matched at this step, then the algorithm succeeds, otherwise the algorithm fails. The following theorem provides the achievable distribution region for this algorithm.

\begin{Theorem}
Let $p_n\in [\omega({{\log{n}}/{n^{\frac{1}{5}}}}),1-\omega({\log{n}}/{n^{\frac{1}{5}}})]$. Then, the MDA is a graph isomorphism algorithm for the CPER $\underline{g}_{n,P_{n,X_1,X_1}}=(\tilde{g}^1_{n,p_n},{g}^2_{n,p_n})$.
\label{th:MDA}
\end{Theorem}
\begin{proof}
This result was shown by Bollob\'as \cite{iso2}. We reprove the theorem using a new method. The tools introduced in the proof are used in the next section when we study seeded graph matching.
Throughout the proof we write $p$ instead of $p_n$ for brevity. As explained above, the matching algorithm succeeds if every vertex has a unique signature. Define the following error event:
\begin{align*}
 \mathcal{D}_n=\{\underline{g}_{n,P_{n,X_1,X_1}}| \exists k\neq k', \underline{F}^j_{n,(i),k}=\underline{F}^j_{n,(i),k'}\}.
\end{align*}
The algorithm succeeds as long as $P(\mathcal{D}_n)\to 0$ as $n\to \infty$. We have:
\begin{align}
 P(\mathcal{D}_n)&= P\left(\bigcup_{\substack{k\neq k' \\ i\in [1,n]}} \{\underline{g}_{n,P_{n,X_1,X_1}}| \underline{F}^j_{n,(i),k}=\underline{F}^j_{n,(i),k'}\}\right)
 \stackrel{(a)}{\leq}\sum_{i=1}^n\sum_{k=1}^{k_{n,i}}\sum_{k'<k}P\left(\underline{F}^j_{n,(i),k}=\underline{F}^j_{n,(i),k'}\right)
 \stackrel{(b)}{\leq} n^2P\left(\underline{F}^j_{n,(i),k}=\underline{F}^j_{n,(i),k'}\right),\label{eq:bound10}
 \end{align}
where in (a) we have used the union bound and (b) follows from the fact that there are a total of $n$ vertices. Next, we calculate $P\left(\underline{F}^j_{n,(i),k}=\underline{F}^j_{n,(i),k'}\right)$:
\begin{align*}
 P&\left(\underline{F}^j_{n,(i),k}=\underline{F}^j_{n,(i),k'}\right)=\sum_{\underline{F}'\subset \{0,1\}^{d_{n,u}}} 
 P(\underline{F}'=\underline{F}^j_{n,(i),k}=\underline{F}^j_{n,(i),k'})
 \stackrel{(a)}{=} \sum_{\underline{F}'\subset \{0,1\}^{d_{n,u}}}  P(\underline{F}'=\underline{F}^j_{n,(i),k})P(\underline{F}'=\underline{F}^j_{n,(i),k'})\nonumber
 \\
 &= \sum_{\underline{F}'\subset \{0,1\}^{d_{n,u}}} p^{2w_H(\underline{F}')}(1-p)^{2(n-w_H(\underline{F}'))}
= \sum_{i=1}^{d_{n,u}} {d_{n,u} \choose i} p^{2i}(1-p)^{2(d_{n,u}-i)}\nonumber\\
 &\stackrel{(b)}{=} \sum_{i=1}^{d_{n,u}} \frac{1}{\sqrt{\lambda d_{n,u}}}2^{nH_b(\frac{i}{d_{n,u}})}\left(1+O(\frac{1}{d_{n,u}})\right) p^{2i}(1-p)^{2(d_{n,u}-i)}\nonumber,
 \end{align*}
 
 where (a) follows from the fact that the difference between subsequent unique degrees is at least 2 with probability one \footnote{For a complete argument on the independence of the edge incidence refer to the proof of Theorem 3.17 in \cite{iso2}.}, and in (b) we have used the following equality. 
\begin{align*}
 {d_{n,u} \choose i}=\frac{1}{\sqrt{\lambda d_{n,u}}}2^{nH_b(\frac{i}{d_{n,u}})}\left(1+O(\frac{1}{d_{n,u}})\right),
\end{align*}
where $\lambda=2\pi p(1-p)$ and $H_b$ is the binary entropy function. The proof of the equality follows from Sterling's approximation $d_{n,u}!=\sqrt{2\pi n} e^{-d_{n,u}}{d_{n,u}}^{d_{n,u}}\left(1+O(\frac{1}{d_{n,u}})\right)$. Following the steps in the proof of Theorem 2 in \cite{Allerton}, we have: 
\begin{align}
\label{eq:bound11}
  P\left(\underline{F}^j_{n,(i),k}=\underline{F}^j_{n,(i),k'}\right)\leq (p^2+(1-p)^2)^{d_{n,u}}.
\end{align}
Combining Equations \eqref{eq:bound10} and \eqref{eq:bound11} gives:
\begin{align}
 P(\mathcal{D}_n)\leq n^2(p^2+(1-p)^2)^{d_{n,u}},
 \end{align}
which goes to zero as $n\to \infty$ if $d_{n,u}=\omega\left(\frac{\log{n}}{H_2(p)}\right)$, where $H_2(p)$ is the R\'enyi entropy with parameter 2 given by:
\begin{align*}
 H_2(p)=\log{\frac{1}{p^2+(1-p)^2}}.
\end{align*}
So, from Lemma \ref{lem:degree}, the algorithm succeeds with probability one as long as:
\begin{align*}
 \frac{\log{n}}{H_2(p)}=o\left(\left(\frac{p(1-p)n}{\log{n}}\right)^{\frac{1}{4}}\right).
\end{align*}
Note that:
 \begin{align*}
 \frac{H_2(p)}{\log_2(e)}&=-\log_e(p^2+(1-p)^2)=\log_e(\frac{1}{p^2+(1-p)^2})
 \stackrel{(a)}{\geq} 1-\frac{1}{\frac{1}{p^2+(1-p)^2}}=1-p^2-(1-p)^2
=2p(1-p)\stackrel{(b)}{\geq} p,
\end{align*}
 where in (a) we have used $\log_e{x}\geq 1-\frac{1}{x}$, and in (b) we have assumed $p<\frac{1}{2}$ for sufficiently large $n$.  
 
Also, $p<\frac{H_2(p)}{\log_2{e}}$ gives  $\frac{\log{n}}{p}>\frac{\log_2{e}\log{n}}{H_2(p)}$. So, if we show that $\frac{\log{n}}{p}=o\left(\left(\frac{p(1-p)n}{\log{n}}\right)^{\frac{1}{4}}\right)$ it follows that  $\frac{\log{n}}{H_{2}(p)}=o\left(\left(\frac{p(1-p)n}{\log{n}}\right)^{\frac{1}{4}}\right)$. We have:

\begin{align*}
&\frac{\log{n}}{p}=o\left(\left(\frac{p(1-p)n}{\log{n}}\right)^{\frac{1}{4}}\right)
\Leftrightarrow n^{-\frac{1}{4}}\log{n}^{\frac{5}{4}}=o(p^{\frac{5}{4}}(1-p)^{\frac{1}{4}})\\
&\Leftrightarrow n^{-\frac{1}{4}}\log{n}^{\frac{5}{4}}=o(p^{\frac{5}{4}}) \Leftrightarrow n^{-\frac{1}{5}}\log{n}=o(p)
\Leftrightarrow p=\omega(n^{-\frac{1}{5}}\log{n}).
\end{align*}
\end{proof}


\section{Seeded Graph Matching}
\begin{figure}
\centering 
\includegraphics[width=4in, draft=false]{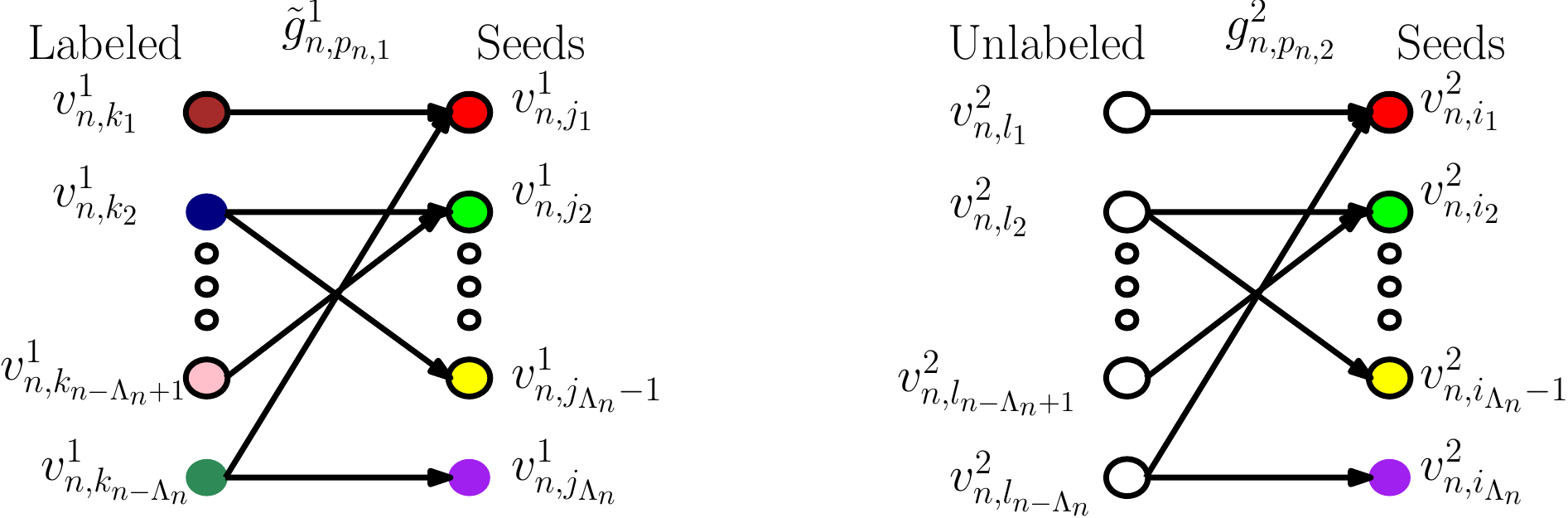}
\caption{The agent constructs the bipartite graph which captures the connections between the unmatched vertices with the seed vertices.}
\label{fig:bipart}
\end{figure}
In this section, we consider the seeded graph matching problem. We argue that seeded graph matching is related to the active fingerprinting problem studied in \cite{Allerton}. In active fingerprinting, the agent has access to two correlated bipartite graphs. The vertices in each graph are partitioned into a set of anonymized users and a set of groups. The objective is to match a single anonymized user in the first bipartite graph to a corresponding user in the second graph based on their connections with group vertices. 

 To demonstrate the connection between seeded graph matching and active fingerprinting, we view seeds in graph matching as groups in active fingerprinting and unmatched vertices as anonymized users. Then. the seeded matching problem can be solved using a variation of the active fingerprinting attacks in \cite{Allerton}. 
The proposed graph matching algorithm operates as follows. First, the agent constructs the bipartite graph shown in Figure~\ref{fig:bipart} whose edges consist of the connections between the unmatched vertices with the seeded vertices in each graph. Next, it matches the two bipartite graphs by iteratively applying the fingerprinting algorithm proposed in \cite{Allerton}. The fingerprinting algorithm operates in two steps. First, the agent constructs signature vectors for each of the unmatched vertices in the two bipartite graphs based on their connections to the seed vertices. In the second step, it uses joint typicality between the signature vectors to match the two vertices. The iterative application of this algorithm means that the agent first uses the fingerprinting algorithm to match a subset of the vertices. Then, the matched vertices in this step are added to the seed set and this expanded set of seeds is used in the next iteration to match the rest of the vertices.

Formally, assume that the agent is to match the CPER $\underline{g}_{n,P_{n,X_1,X_2}}=(\tilde{g}^1_{n,p_{n,1}},{g}^2_{n,p_{n,2}})$
 using the seed set $\mathcal{S}_n\subset \mathcal{V}^2_n$, and let $\epsilon_n= \omega(\frac{1}{\sqrt{\Lambda_n}})$.  
  Let $\mathcal{S}_n=\{v^2_{n,i_1},v^2_{n,i_2},\cdots, v^{2}_{n,i_{\Lambda_n}}\}$ and define the reverse seed set  $\mathcal{S}^{-1}_n=\{v^1_{n,j_1},v^1_{n,j_2},\cdots, v^{1}_{n,j_{\Lambda_n}}\}$, where $\sigma^1_n(v^1_{n,j_k})=\sigma^2_n(v^2_{n,i_k}), k\in [1,\Lambda_n]$.
 The agent is given $\sigma^1_n:\mathcal{V}^1_n\to [1,n]$ as well as $\sigma^2_n|_{\mathcal{S}_n}: \mathcal{S}_n\to [1,n]$. The objective is to find $\hat{\sigma}^2_n: \mathcal{V}^2_n\to [1,n]$ such that $P(\hat{\sigma}^2_n=\sigma^2_n)\to 1$ as $n\to \infty$. To this end, the agent first constructs a signature for each vertex in each of the graphs. For an arbitrary vertex $v^1_{n,i}$ in $g^1_{n,p_{n,1}}$, its signature is defined as the binary vector  $\underline{F}^1_{n,j}=(F^1_{n,j}(1), F^1_{n,j}(2),\cdots,F^1_{n,j}(\Lambda_n))$.
which indicates its connections to the reverse seed elements:
\begin{align*}
{F}^1_{n,j}(l)=
\begin{cases}
 1 \qquad & \text{if} \qquad (v^1_{n,j}, v^1_{n,j_l})\in \mathcal{E}^1_{n}\\
 0& \text{Otherwise}
\end{cases},
\qquad l\in [1,\Lambda_n].
\end{align*}
The signature of an arbitrary vertex $v^2_{n,i}$ is defined in a similar fashion based on connections to the elements of the seed set $\mathcal{S}_n$. Take an unmatched vertex $v^2_{n,i}\notin \mathcal{S}_n$. The agent matches $v^2_{n,i}$ to a vertex $v^1_{n,j}$ if it is the unique vertex such that the signature pair $(\underline{F}^1_{n,j}, \underline{F}^2_{n,i})$ are jointly $\epsilon$-typical with respect to the distribution $P_{n,X_1,X_2}$. Alternatively, 
\begin{align*}
 \exists ! i: (\underline{F}^1_{n,j}, \underline{F}^2_{n,i})\in A_{\epsilon}^n(X_1,X_2) \Rightarrow \hat{\sigma}^2_n(v^2_{n,i})=\sigma^1(v^1_{n,j}),
\end{align*}
where $A_{\epsilon}^n(X_1,X_2)$ is the set of jointly $\epsilon$-typical set sequences of length $n$ with respect to $P_{n,X_1,X_2}$. If a unique match is not found, then vertex $v^2_{n,i}$ is added to the ambiguity set $\mathcal{L}^2_n$. Assume that at the end of this step, the set of all matched vertices is  $\mathcal{V}^2_n\backslash \mathcal{L}^2_n$. In the next step, these vertices are added to the seed set and the expanded seed set is used to match the vertices in the ambiguity set.
The algorithm succeeds if all vertices are matched at this step and fails otherwise. We call this strategy the Typicality Matching Strategy (TMS).

\begin{Theorem}
Define the family of sets of pairs of distribution and seed sizes $\widetilde{\mathcal{P}}$ as follows:
\begin{align*}
 \widetilde{\mathcal{P}}= \{(\mathcal{P}_n,\Lambda_n)_{n\in \mathbb{N}}| \forall P_{n,X_1,X_2}\in \mathcal{P}_n: \frac{2\log{n}}{I(X_1,X_2)}\leq \Lambda_n \}.
\end{align*}
Then, $\widetilde{\mathcal{P}}$ is achievable using the TMS.
\label{th:seeded}
\end{Theorem}
To prove Theorem \ref{th:seeded}, we need the following lemma on the cardinality of $\mathcal{L}^2_n$:
 \begin{Lemma}
 The following holds:
 \begin{align*}
  P(|\mathcal{L}^2_n|> \frac{2n}{\Lambda_n\epsilon^2})\to 0, \text{ as } n\to \infty,
\end{align*}
\end{Lemma}
\begin{proof}
Formally, $\mathcal{L}^2_n$ is defined as:
\begin{align*}
 \mathcal{L}^2_n=\{v^2_{n,i}\big| \nexists! j: (\underline{F}^1_{n,j}, \underline{F}^2_{n,i})\in A_{\epsilon}^n(X_1,X_2)\}.
\end{align*}
From the Chebychev inequality, we have:
 \begin{align}
 P\left(|\mathcal{L}^2_n|>2\mathbb{E}(|\mathcal{L}^2_n|)\right)=  P\left(\big||\mathcal{L}^2_n|-\mathbb{E}(|\mathcal{L}^2_n|)\big|>\mathbb{E}(|\mathcal{L}^2_n|)\right)
 \leq \frac{Var(|\mathcal{L}^2_n|)}{\mathbb{E}^2(|\mathcal{L}^2_n|)}.
 \label{eq:cheb}
\end{align}

 Let $A_{j}$ be the event that $v^2_{n,j}\in \mathcal{L}^2_n$, then
\begin{align*}
&\mathbb{E}(|\mathcal{L}^2_n|)=
\mathbb{E}\left(\sum_{j=1}^n\mathbbm{1}(v^2_{n,j}\in \mathcal{L}^2_n)\right)=
\sum_{j=1}^nP(v^2_{n,j}\in \mathcal{L}^2_n)=\sum_{j=1}^n P(A_j)\\
 &Var(|\mathcal{L}^2_n|)= \sum_{j=1}^n P(A_j)+ \sum_{i\neq j} P(A_i,A_j)- \left(\sum_{j=1}^n P(A_j)\right)^2\\
 &= \sum_{j=1}^n P(A_j)- \sum_{j=1}^n P^2(A_j)\leq \mathbb{E}(|\mathcal{L}^2_n|).
\end{align*}
So, from \eqref{eq:cheb}, we have:
\begin{align*}
 P\left(|\mathcal{L}^2_n|>2\mathbb{E}(|\mathcal{L}^2_n|)\right)\leq \frac{1}{\mathbb{E}(|\mathcal{L}^2_n|)},
\end{align*}
which goes to 0 as $n\to \infty$ provided that $\mathbb{E}(|\mathcal{L}^2_n|) \to \infty$ (otherwise the claim is proved since $\mathbb{E}(|\mathcal{L}^2_n|)$ is finite.).
It remains to find an upper bound on $\mathbb{E}(|\mathcal{L}^2_n|)$. Let $B_j$ be the event that the signature $\underline{F}^2_{n,j}$ is not typical with respect to $P_{n,X_2}$ and let $C_{i,j}$ be the event that there exists $i\in [1,n]$ such that $\underline{F}^1_{n,j}$ and $\underline{F}^2_{n,i}$ are jointly typical with respect to $P_{n,X_1,X_2}$. Then,
\begin{align*}
 P(A_j)&\leq P\left(B_j\bigcup \left( \bigcup_{i\neq j} C_{i,j}\right)\right)\leq P(B_j)+\sum_{i\neq j} P(C_{i,j}|B_j^c)
 \stackrel{(a)}\leq \frac{1}{\Lambda_s\epsilon^2}+ n2^{-\Lambda_s(I(X_1;X_2)-\epsilon)},
\end{align*}
where (a) follows from the standard information theoretic arguments (e.g proof of Theorem 3 in \cite{Allerton}). So,
\begin{align*}
 \mathbb{E}(|\mathcal{L}^2_n|)\leq \frac{n}{\Lambda_s\epsilon^2}+ n^22^{-\Lambda_s(I(X_1;X_2)-\epsilon)}.
\end{align*}
 From $\Lambda_s>\frac{2\log{n}}{I(X_1;X_2)}$, we conclude that the second term approaches 0 as $n\to \infty$. This completes the proof of Claim 1.
\end{proof}
We proceed to prove Theorem \ref{th:seeded}:
\begin{proof}
Let $H$ be the event the algorithm fails and $K$ the event that $|\mathcal{L}^2_n|> \frac{2n}{\Lambda_n\epsilon^2}$. Then:\[
 P(H)\leq P(K)+P(H|K^c).\] From Claim 1, we know that $P(K)\to 0$ as $n\to \infty$. For the second term, let $\mathcal{L}'_n$ be the set of vertices which are not matched in the second iteration. The algorithm fails if  $\mathcal{L}'_n\neq \phi$. However, by a similar argument as in the proof of claim 1, we have:
 \begin{align*}
 P(|\mathcal{L}'_n|>\frac{1}{2} \Big| |\mathcal{L}^2_n|<\frac{2n}{\Lambda_n\epsilon^2})\to 0 \text{ as } n\to\infty.
\end{align*}
So,  $P(|\mathcal{L}'_n|=0) \to 1$ as $n\to \infty$. This completes the proof.

\end{proof}
\section{Conclusion}
We have introduced a new information theoretic framework for analyzing graph matching problems. We have used this framework to study both  graph isomorphism as well as seeded graph matching. In graph isomorphism, we have analyzed the performance of the maximum degree algorithm and rederived suffiecient conditions for successful matching under the Erd\"os-R\'enyi model. Furthermore, we have proposed a new seeded graph matching algorithm. The new algorithm operates using a method called typicality matching. We have analyzed the performance of the new seeded matching algorithms and derived sufficient conditions on the seed size as functions of graph statistics.

\bibliographystyle{unsrt}

\end{document}